\documentclass{article}

\usepackage{hyperref}
\usepackage{amsthm}
\usepackage{mathtools}
\mathtoolsset{centercolon}
\usepackage[noline,noend,linesnumbered]{algorithm2e}

\newtheorem{lemma}{Lemma}[section]
\newtheorem{proposition}{Proposition}[section]
\newtheorem{theorem}{Theorem}[section]
\newtheorem{example}{Example}[section]
\newtheorem{definition}{Definition}[section]

\usepackage{latexsym}

\mathcode`.="613A          
\mathcode`|="326A          

\newcommand{\defeq}{\stackrel{.}{=}}

\let\implies=\Rightarrow

\let\pleq=\sqleq          

\def\land{\mathrel{\wedge}}

\def\preleq{\mathrel{
  \raise 2pt\hbox{$\mathop\sqsubset\limits_{\hbox{$\sim$}}$}%
}}

\let\union=\cup

\let\join=\sqcup

\newcommand\setsize[1]{\left| #1 \right|}

\newcommand\domainsub{\mathop{\lhd\mkern-14mu-}}

\newcommand\auxfun[1]{\expandafter\newcommand\csname #1\endcsname{%
 \mathop{\hbox{$\mathsf{#1}$}}\nolimits}}

\newcommand\af[1]{\mathop{\hbox{$\mathsf{#1}$}}\nolimits}
\newcommand\kw[1]{\mathrel{\hbox{$\mathbf{#1}$}}}

\auxfun{Max}
\auxfun{fix}
\auxfun{dom}
\auxfun{fst}
\auxfun{snd}
\auxfun{ran}

\auxfun{incr}
\auxfun{fetch}
\auxfun{init}
\auxfun{merge}
\auxfun{receive}
\auxfun{random}
\auxfun{send}

\auxfun{id}
\auxfun{tier}
\auxfun{val}
\auxfun{below}
\auxfun{sck}
\auxfun{dck}
\auxfun{slots}
\auxfun{tokens}
\auxfun{vals}
\auxfun{fn}
\auxfun{discardtokens}
\auxfun{discardslot}
\auxfun{fillslots}
\auxfun{mergevectors}
\auxfun{aggregate}
\auxfun{createslot}
\auxfun{createtoken}
\auxfun{cachetokens}

\newcommand{\mergewith}[3]{\union^{#1}(#2,#3)}

\title{Scalable Eventually Consistent Counters over Unreliable Networks}
\author{Paulo S\'ergio Almeida, Carlos Baquero
\\\{psa,cbm\}@di.uminho.pt
\\HASLab, INESC Tec \& Universidade do Minho
}

\begin{document}

\maketitle

\begin{abstract}

Counters are an important abstraction in distributed computing, and play a
central role in large scale geo-replicated systems, counting events such as
web page impressions or social network ``likes''.
Classic distributed counters, strongly consistent, cannot be made both
available and partition-tolerant, due to the CAP Theorem, being unsuitable to
large scale scenarios.
This paper defines Eventually Consistent Distributed Counters (ECDC) and
presents an implementation of the concept, Handoff Counters, that is scalable
and works over unreliable networks. By giving up the sequencer aspect of
classic distributed counters, ECDC implementations can be made AP in the CAP
design space, while retaining the essence of counting.
Handoff Counters are the first CRDT (Conflict-free Replicated Data Type) based
mechanism that overcomes the identity explosion problem in naive CRDTs,
such as G-Counters (where state size is linear in the number of independent
actors that ever incremented the counter), by managing identities towards
avoiding global propagation and garbage collecting temporary entries.
The approach used in Handoff Counters is not restricted to counters, being
more generally applicable to other data types with associative and commutative
operations.

\end{abstract}

\section{Introduction}

A counter is one of the most basic and important abstractions in computing.
From the small-scale use of counter variables in building data-types, to
large-scale distributed uses for counting events such as web page impressions,
banner clicks or social network ``likes''.
Even in a centralized setting, the increment operation on a counter is
problematic under concurrency, being one of the examples most used to
illustrate the problems that arise if a load, add one and store are
not atomic.
In a distributed setting things are much worse, due to the absence of shared
memory, possibly unreliable communication (message loss, reordering or
duplication), network partitions or node failures.

If one has a strongly consistent distributed database with support for
distributed transactions counters can be trivially obtained.
Unfortunately, such databases are not appropriate for large-scale
environments with wide-area replication, high latency and possible network
partitions. A naive counter obtained through a ``get, add one, and put''
transaction will not scale performance-wise to a wide-area deployment with
many thousands of clients.

The CAP theorem~\cite{Brewer2000, GilbertLynch2002} says that one cannot have
Consistency, Availability, and Partition-tolerance together; one must choose
at most two of these three properties. Therefore, to have an always-available
service under the possibility of partitions (that in world-wide scenarios are
bound to happen from time to time), distributed data stores such as
Dynamo~\cite{Dynamo}, Cassandra~\cite{Cassandra} and Riak~\cite{Klophaus2010}
have been increasingly choosing to go with AP (availability and
partition-tolerance) and give up strong consistency and general distributed
transactions, in what has become known as the NoSQL movement.

With no support for strong consistency and distributed transactions, and an
API mostly based on simple get and put operations, obtaining a ``simple''
counter becomes a problem. NoSQL data stores like Cassandra have been trying
to offer counters natively with an increment operation in the API, something
that has revealed a non-trivial problem, involving many ad hoc evolutions.
The current state in Cassandra is well summarized by~\cite{CASSANDRA4775}:
``The existing partitioned counters remain a source of frustration for most
users almost two years after being introduced.  The remaining problems are
inherent in the design, not something that can be fixed given enough
time/eyeballs.''

An approach towards obtaining provably correct eventually consistent
implementations of data types such as counters are the so called CRDTs:
Conflict-free Replicated Data Types~\cite{Shapiro2011CRDT}. The idea is that
each node keeps a CRDT (a replica) that can be locally queried or operated upon,
giving availability even under partitions, but providing only eventual
consistency: queries can return stale values, but if ``enough'' messages go
through, then all nodes will converge to the correct value. CRDT-based data
types (the state-based ones) can be built to trivially tolerate unreliable
communication.  They are designed so that duplicates or non-FIFO communication
are not a problem.

CRDTs are not, however, the silver bullet they are sometimes assumed to be. A
problem that can easily arise in CRDTs is scalability. This problem is easy to
explain: the CRDT approach envisions a CRDT per participating entity; in many
CRDTs each entity needs to have a unique identity, and many CRDTs are made to
work by keeping maps from participating entities ids to some value; the set of
keys in these maps will keep growing along time, as more entities participate,
preventing scalability.

This means that practical CRDT-based approaches to counters involving some
server nodes and possibly many clients will exclude clients from the entities
having CRDTs, having them server-side only. This will solve the scalability
problem and allow unreliable communication between servers, but will not solve
the fault-tolerance problem in the client-server interaction. 
This is because a basic problem with counters is that the increment operation
is not idempotent; therefore, an increment request by a client (which itself
does not keep a CRDT) cannot just be re-sent to the server in case there is
no acknowledgment. This problem is well recognized by practitioners, as can
be seen by, e.g.,~\cite{CASSANDRA2495}.

If one looks at theory of distributed counting, a substantial
amount of work has been done, namely
Software Combining Trees~\cite{YewTzengLawrie1987,GoodmanVernonWoest1989},
Counting Networks~\cite{AspnesHerlihyShavit1994}, 
Diffracting Trees~\cite{ShavitZemach1996} and
Counting Pyramid~\cite{WattenhoferWidmayer2004} (and many variants of these,
specially of counting networks).
However, all these works address a strongly consistent definition of
counter, as a data type that provides a single ``fetch-and-increment''
operation in the sense of~\cite{Stone1984}. Although an unquestionably
powerful abstraction (e.g., to generate globally unique sequence numbers), it
is indeed too powerful to be implemented while providing availability under
unreliable communication with possible partitions. The focus of these works
is scalability (mainly avoiding contention or bottlenecks) and not
fault-tolerance. While some aspects like wait-freedom are addressed, to
tolerate failures of processes, tolerance to message loss or component failure
(e.g., a balancer in the case of counting networks) is not addressed. This
means that the applicability of these mechanisms is mostly in tightly-coupled,
low-latency, failure-free environments, such as multiprocessors, serving as
scalable alternatives to lock-based counter implementations, as discussed
in~\cite{HerlihyLimShavit1995}.

The remainder of this paper is organized as follows:
Section~\ref{classic_counters} briefly revisits classic strongly consistent
distributed counters, and explains why they are not suitable for large-scale
AP scenarios. In Section~\ref{ec_counters} we propose a definition of
\emph{Eventually Consistent Distributed Counters} (ECDC), stating both safety
and liveness conditions.
Section~\ref{naive_crdt_counters} describes a naive eventually consistent
CRDT-based counter, that is AP in unreliable networks, and explains its
scalability problems; it also discusses in more detail the problems that arise
if these CRDTs are restricted to server-side.  In
Section~\ref{handoff_counters} we present \emph{Handoff Counters}, our new
CRDT-based mechanism that implements eventually consistent counters that is
simultaneously reliable, available and partition-tolerant under unreliable
networks, and scalable in the number of entities (both active and already
terminated).  Section~\ref{correctness} contains formal correctness proofs for
the mechanism.  In Section~\ref{practical_considerations} we address some
practical implementation issues. In Section~\ref{beyond-counters} we discuss
how the handoff mechanism proposed can be applied to more general scenarios,
beyond simple counters, to commutative monoids having an associative and
commutative operation, and conclude in Section~\ref{conclusion}.

\section{Classic Distributed Counters}
\label{classic_counters}

In most papers about distributed counters, e.g.,~\cite{AspnesHerlihyShavit1994,
ShavitZemach1996, WattenhoferWidmayer2004}, a counter is an abstract data type
that provides a fetch-and-increment operation (increment for short), which
returns the counter value and increments it. The basic correctness criteria is
usually: a counter in a quiescent state (when no operation is in progress)
after $n$ increments have been issued will have returned all values from $0$
to $n-1$ with no value missing or returned twice. I.e., when reaching a
quiescent state, all operations so far must have behaved as if they have
occurred in some sequential order, what is known as \emph{quiescent
consistency}~\cite{AspnesHerlihyShavit1994}. Some counter mechanisms,
e.g.,~\cite{WattenhoferWidmayer2004}, 
enforce the stronger \emph{linearizability}~\cite{HerlihyWing1990} condition,
which ensures that, whenever a first increment returns before a second one is
issued, the first returns a lower value than the second.

Even forgetting the stronger variants that enforce linearizability, classic
distributed counters providing quiescent consistency are too strongly
consistent if one is aiming for availability and partition tolerance.
For classic distributed counters we have the following result for
deterministic algorithms, with a trivial proof, which is nothing more than an
instantiation of the CAP theorem:

\begin{proposition}
A quiescently consistent fetch-and-increment counter cannot
be both available and partition tolerant.
\end{proposition}

\begin{proof}
Suppose a network with two nodes $u$ and $v$, and a run $A$ where they are
partitioned. Assume an increment is issued at node $u$ at
time $t_1$ and no other operations are in progress. As the counter is
available and partition tolerant, it will eventually return at some later time
$t_2$. Because the system is in a quiescent state after $t_2$, this increment
must have returned $0$. Suppose an increment is then issued at node $v$ at
some later time $t_3 > t_2$. For the same reasons, this increment will
eventually return, the system becomes quiescent again, and the returned value
must, therefore, be $1$.  But as no messages got through between $u$ and $v$,
this run is indistinguishable by $v$ from a run $B$ in which $u$ does not exist
and only a single increment is issued by $v$. In run $B$, $v$ will,
therefore, behave the same as in run $A$ and return the same value $1$, which
contradicts the requirement for run $B$ that a single increment in the whole
run should have returned $0$ after reaching quiescence.
\end{proof}

This simply formalizes the intuition that it is not possible to generate
globally unique numbers forming a sequence in a distributed fashion without
communication between the nodes involved.

\section{Eventually Consistent Distributed Counters}
\label{ec_counters}

Given the too strongly consistent nature of classic distributed counters, to
achieve availability and partition tolerance a weaker variant of distributed
counters needs to be devised. In this section we define such a variant, that
we call \emph{eventually consistent distributed counters} (ECDC).

Classic counters offer a two-in-one combination of two different features: 1)
keeping track of how many increments have been issued; 2) returning globally
unique values. While undoubtedly powerful, this second feature is the
problematic one if aiming for AP.

For many practical uses of counters (in fact what is being offered in NoSQL
data stores like Cassandra and Riak) one can get away with not having the
second feature, and having a counter as a data type that can be used to count
events, by way of an \emph{increment} operation, which does not return
anything, and an independent \emph{fetch} operation, which returns the value of
the counter. This splitting makes clear that one is not aiming for obtaining
globally unique values (fetch can return the same value several times), and
one can start talking about possibly returning stale values.
Having two independent operations, one to mutate and the other to report, as
opposed to a single atomic fetch-and-increment, corresponds also to the more
mundane conception of a counter, and to what is required in a vast number of
large scale practical uses, where many participants increment a counter while
others (typically less, usually different) ask for reports.

It will be possible to obtain available and partition tolerant eventually
consistent counters, by allowing fetch to return something other than the more
up-to-date value.
Nevertheless, we need concrete correctness criteria for ECDC. We have devised
three conditions. Informally:

\begin{itemize}
\item A fetch cannot return a value greater than the number of
increments issued so far.
\item At a given node, a fetch should return at least the sum of the value
returned by the previous fetch (or 0 if no such fetch was issued)
plus the number of increments issued by this node between these two fetches.
\item All increments issued up to a given time will be reported eventually at
a later time (when the network has allowed enough messages to go through).
\end{itemize}

The first two criteria can be thought of as safety conditions. The first, not
over-counting, is the more obvious one (and also occurs in classic distributed
counters, as implied by their definition).  The second is a local condition on
session guarantees \cite{Terry:1994:SGW:645792.668302}, analogous to having
\emph{read-your-writes} and \emph{monotonic-read}, common criteria in eventual
consistency~\cite{Vogels2009}.  The third is a liveness condition, which states
that eventually, if network communication allows, all updates are propagated and
end up being reported. It implies namely that if increments stop being issued,
eventually fetch will report the correct counter value, i.e., the number of
increments. We will now clarify the system model, and subsequently formalize
the above correctness criteria for ECDC.

\subsection{System Model}
\label{system-model}

Consider a distributed system with nodes containing local memory, with no
shared memory between them.
Any node can send messages to any other node. The network is asynchronous,
there being no global clock, no bound on the time it takes for a message to
arrive, nor bounds on relative processing speeds. The network is unreliable:
messages can be lost, duplicated or reordered (but are not corrupted). Some
messages will, however, eventually get through: if a node sends infinitely
many messages to another node, infinitely many of these will be delivered. In
particular, this means that there can be arbitrarily long partitions, but
these will eventually heal.

Nodes have access to stable storage; nodes can crash but eventually will
recover with the content of the stable storage as at the time of the crash.
Each node has access to a globally unique identifier.

As we never require that data type operations block waiting for other
operations or for message reception, they are modeled as single atomic
actions.
(In I/O Automata~\cite{LynchTuttle1987} parlance, we will use a single action
as opposed to a pair $\af{opStart}$ (input action), and $\af{opEnd}(r)$
(output action) ``returning'' $r$).
This allows us to use $\af{op}_i^t$ to mean that operation op was performed by
node $i$ at time $t$, and in the case of fetch also for the result of that
operation. The actions we use are $\fetch_i$ and $\incr_i$ for the data type
operations, and $\send_{i,j}(m)$ and $\receive_{i,j}(m)$ for message exchange.

\subsection{Formal Correctness Criteria for ECDC}

An eventually consistent distributed counter is a distributed abstract data type
where each node can perform operations $\fetch$ (returning an integer) and
$\incr$ (short for increment), such that the following conditions hold (where
$\setsize{\ }$ denotes set cardinality and $\_$ the unbound variable, matching
any node identifier; also, for presentation purposes, we assume an implicit
$\fetch_\_^0 = 0$ at time 0 by all nodes). For any node $i$, and times $t, t_1,
t_2$, with $t_1 < t_2$:

\paragraph{Fetch bounded by increments:}

\[\fetch_i^t \leq \setsize{\{\incr_\_^{t'} | t' < t\}},\]

\paragraph{Local monotonicity:}


\[\fetch_i^{t_2} - \fetch_i^{t_1} \geq
\setsize{\{\incr_i^{t'} | t_1 < t' < t_2\}},\]

\paragraph{Eventual accounting:}

\[\exists t' \geq t . \forall j. \fetch_j^{t'} \geq
\setsize{\{\incr_\_^{t''} | t'' < t\}}.\]

These criteria, specific to counters, can be transposed to more general
consistency criteria, namely they imply the analogous for counters of:

\begin{description}
\item[Eventual Consistency] From \cite{Vogels2009} \emph{``[It] guarantees
that if no new updates are made to the object, eventually all accesses will
return the last updated value.''}. Eventual accounting is stronger than
eventual consistency: it does not require increments to stop, but clearly
leads to eventual consistency if increments do stop. All CRDTs include this
consistency criteria
\cite{Shapiro2011CRDT}. 

\item[Read-your-writes] From \cite{Terry:1994:SGW:645792.668302} \emph{``[It]
ensures that the effects of any Writes made within a session are visible to
Reads within that session.''}. The analogous of this property, substituting
increments for writes, is implied by local monotonicity: in a session where a
process issues increments to a given node, at least the effect of those
increments is seen by further fetches by that process.

\item[Monotonic-reads] Defined in \cite{Terry:1994:SGW:645792.668302} and using
the formulation in \cite{Vogels2009} \emph{``If a process has seen a particular
value for the object, any subsequent accesses will never return any previous
values.''}. This property is also obtained by local monotonicity.
\end{description}

\section{Naive CRDT-based Counters}
\label{naive_crdt_counters}

A state-based CRDT amounts to a replica that can be locally queried or
updated; it directly provides availability and partition tolerance, as all
client-visible data type operations are performed locally, with no need for
communication.  Information is propagated asynchronously, by sending the local
state to other nodes (e.g., using some form of gossip~\cite{Gossip1987}); the
receiving node performs a \emph{merge} between the local and the received
CRDT. CRDTs are designed so that: their abstract state forms a join
semilattice (a partially ordered set for which there is a defined least upper
bound for any two elements, see, e.g.,~\cite{davey2002}); the merge operation
amounts to performing a mathematical \emph{join} of the correspondent abstract
states, deriving their least upper bound; every data type operation is an
\emph{inflation} that moves the state to a larger value (i.e., $\forall x.
f(x) \geq x$).  This means that merges can be performed using arbitrary
communication patterns: join is associative, commutative and idempotent;
duplicates are, therefore, not a problem and in doubt of message loss, a
message can be resent and possibly re-merged; old messages received
out-of-order are also not a problem.  CRDTs solve, therefore, the problem of
unreliable communication for data types that conform to their design.

The CRDT concept can be used to trivially obtain an ECDC. The local state will
amount to a version-vector~\cite{ParkerPopek1983}: a map of node ids to
non-negative integers. When a node wants to perform an increment, it
increments the entry corresponding to its unique identifier (or adds an entry
mapped to one if the id is not mapped). The fetch is obtained by adding all
integers in the map.  The merge operation also corresponds to reconciliation
of version-vectors: maps are merged by performing a pointwise maximum (each
key becomes mapped to the maximum of the corresponding values, assuming absent
keys are implicitly mapped to 0).

It is easy to see that these version-vector based counters respect the
criteria for ECDC. The fetch results from adding values that result from
increments performed, being bounded by the number of increments; local
increments are immediately accounted; as increments by different nodes are
performed in disjoint entries, if each CRDT is propagated and merged to every
other one, all will converge to a CRDT that stores the exact number of
increments performed by each node in the corresponding map entry; therefore,
all increments are eventually accounted.

\subsection{The Scalability Problem of Client-side CRDTs}

Counters implemented as version-vectors, although meeting all criteria for
ECDC over unreliable networks, suffer from serious scalability
problems. Consider a network in which many nodes (clients) perform
operations, while others (servers) allow information propagation and keep
durable storage after client nodes have ceased from participating.

The pure CRDT approach assumes that all participating entities have a CRDT. In
this case, each participating node (both clients and servers) will introduce
its id in the map. Over time the map will grow to unreasonable sizes, making
both the storage space and communication costs (of transmitting the map to be
merged on the other node) unbearable. The worst aspect is that the size does
not depend only on the currently participating clients: it keeps growing,
accumulating all ids from past clients that have already stopped
participating. This means that naive client based CRDTs are not scalable and
not usable in some relevant practical cases.

\subsection{The Availability Problem of Server-side CRDTs}

Due to the above scalability problem, current version-vector based counters
(e.g., in Cassandra or Riak) do not use the pure CRDT approach, but use CRDTs
server-side only. This means that only a relatively small number of nodes (the
servers) hold CRDTs, while clients use remote invocations to ask a server to
perform the operation. Server-side CRDTs allow unreliable communication
between servers, including partitions (e.g., between data-centers).
However, the problem of unreliable communication between client and server
remains.

As the increment operation is not idempotent, it cannot be simply reissued if
there is doubt whether it was successful. In practice, this leads to the use
of remote invocations over connection-oriented protocols (e.g.,
TCP~\cite{CerfKahn1974}) that provide reliable communication. This only
partially solves the problem: a sequence of acknowledged increments is known
to have been applied exactly-once, but if the last increment is not
acknowledged and the connection timeouts, this last increment is not known to
have been successfully applied, but cannot be reissued using a new
connection, to the same or a different server, as it could lead to
over-counting.

Attempts to circumvent this reliability problem through a general
data-type-agnostic communication layer bring back scalability and/or
availability problems. If an infinite duration connection incarnation is
maintained for each client, where no operation can fail due to a timeout, this
will imply stable server state to manage each client connection, leading to
state explosion in servers, as clients cannot be forgotten. This because there
is no protocol that gives reliable message transfer within an infinite
incarnation for general unbounded capacity (e.g., wide-area networks) non-FIFO
lossy networks that does not need stable storage between
crashes~\cite{AttiyaDolevWelch1995,FeketeLynchMansourSpinelli1993}.  This
problem can be overcome by never failing due to a timeout, but allowing
connections to be closed if there are no pending requests and the connection
close handshake is performed successfully (e.g., before a partition occurs).
With a three-way handshake an oblivious protocol is
possible~\cite{AttiyaRappoport1994}, with no need to retain connection
specific information between incarnations, but only a single unbounded counter
for the whole server.

With this last approach, the size of stable server state is not a problem in
practice, but the reliability problem is overcome at the cost of availability:
given a partition, a pending request will never fail due to a time-out, but the
client that has issued a request to a given server will be forced to wait
unboundedly for the result from that server, without being able
to give up waiting and continue the operation using a different server.

These problems can be summarized as: the use of a general purpose
communication mechanism to send non-idempotent requests can provide
reliability at the cost of availability. Our data-type specific mechanism
overcomes this problem by allowing client-side CRDTs which are scalable.

\section{Handoff Counters}
\label{handoff_counters}

In this section we present a novel CRDT based counter mechanism, which we call
\emph{Handoff Counters}, that meets the ECDC criteria, works in unreliable
networks and, as opposed to simple version vector CRDT counters, is scalable.
The mechanism allows arbitrary numbers of nodes to participate and adopts the
CRDTs everywhere philosophy, without distinguishing clients and servers,
allowing an operation (fetch or increment) to be issued at any node.

It addresses the scalability issues (namely the id explosion in maps) by:
assigning a tier number (a non negative integer) to each node; promoting an
hierarchical structure, where only a small number of nodes are classified as tier
0; having ``permanent'' version vector entries only in (and for) tier 0 nodes,
therefore, with a small number of entries; having a \emph{handoff} mechanism
which allows a tier $n+1$ ``client'' to handoff values to some tier $n$
``server'' (or to any smaller tier node, in general); making the entries
corresponding to ``client'' ids be garbage-collected when the handoff is
complete. Figure \ref{fig:simpleTiers} illustrates a simple configuration, with
end-client nodes connecting to tier 1 nodes in their regional datacenters.

\begin{example} 
\label{ex:datacenter} 
Even though no formal client/server
distinction is made, a typical deployment scenario would be having, e.g., two
(for redundancy) tier 0 nodes per data-center, devoted to inter-datacenter
communication, a substantial, possibly variable, number of tier 1 server nodes
per datacenter, and a very large number of tier 2 nodes in each datacenter. The
datastore infrastructure would be made up of tier 0 and 1 nodes, while tier 2
nodes would be the end-clients (application server threads handling end-client
connections or code running at the end-clients) connecting to tier 1 nodes in a
datacenter. More tiers can be added in extreme
cases, but this setup will be enough for most purposes: e.g., considering 5
datacenters, 50 tier 1 nodes per datacenter, each serving 1000 concurrent tier 2
clients, will allow 250000 concurrent clients; in this case the ``permanent''
version vectors will have 10 entries.  
\end{example}

\begin{figure}
\centering
\includegraphics[width=0.7\textwidth]{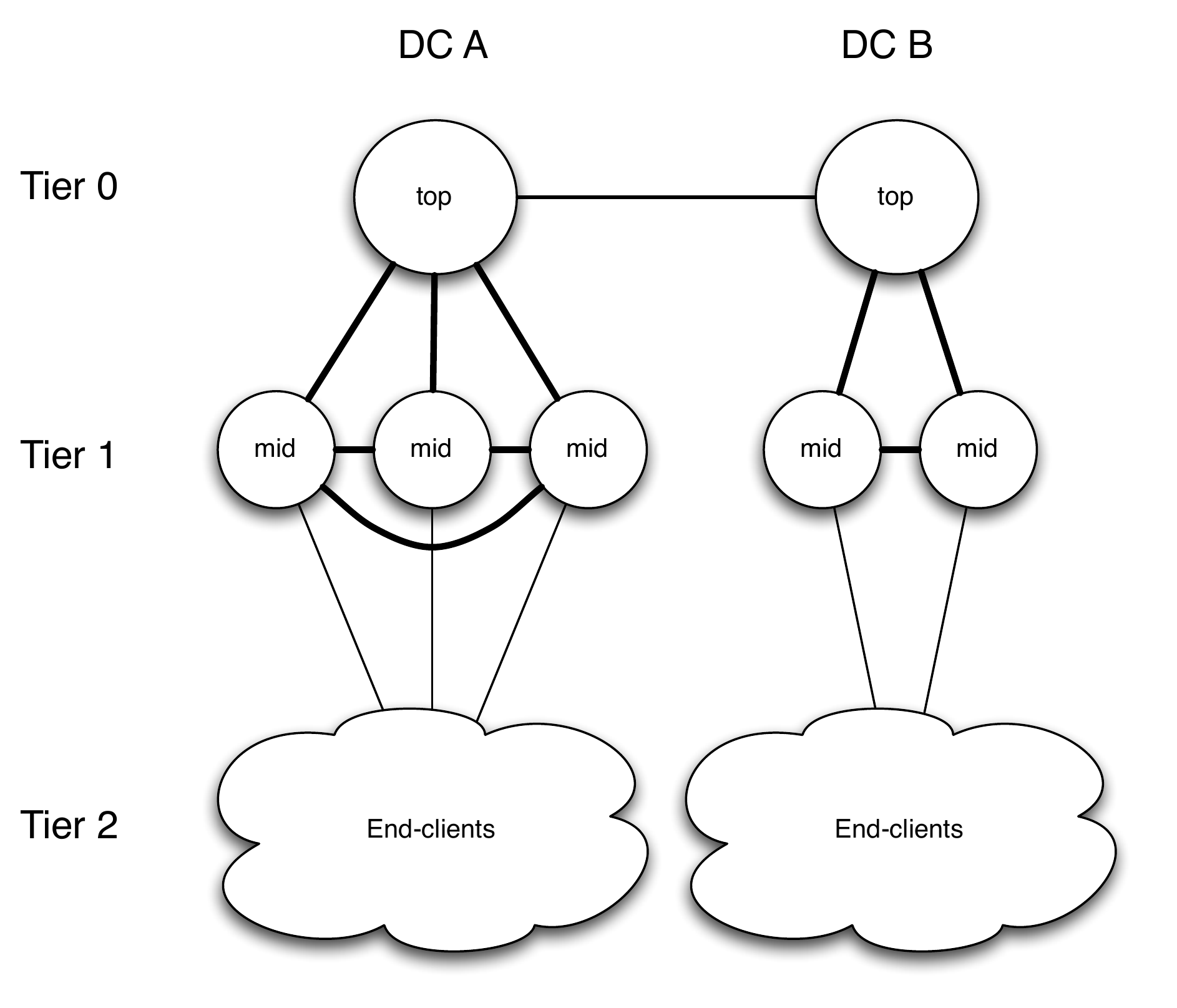}
\caption{A simple configuration with three tiers and two datacenters.}
\label{fig:simpleTiers}
\end{figure}

The mechanism involves three aspects: the handoff, which migrates accounted
values towards smaller tiers, making all increments eventually be accounted
in tier 0 nodes; classic version vector dissemination and merging between tier
0 nodes; top-down aggregation from tier 0 to bottom-tier nodes, to provide a
best-effort monotonic estimate of the counter value.

Most complexity of the mechanism is related to achieving the handoff without
violating correctness under any circumstances over unreliable networks, while
allowing garbage-collection in typical runs. Other design aspects were
considered, in addition to strictly achieving ECDC, namely:

\begin{itemize}
\item an end-client (maximum tier node) is typically transient; it should be
able to know if locally issued increments have already been handed-off, so
that it can stop the interaction and terminate;
\item there should be no notion of session or affinity; a tier $n+1$ node $u$
that started exchanging messages with a tier $n$ node $v$, should be able to
switch to another tier $n$ node $w$ at any time, e.g., if no message arrives
and $u$ suspects that $v$ has crashed or there is a network partition between
$u$ and $v$, but $u$ is an end-client that wants to be sure its locally
accounted increments have been handed-off to some server before terminating.
\end{itemize}

\subsection{Network Topology}
\label{network-topology}

Handoff counters can be used with many different network topologies. The
simplest one is to assume a fully connected graph, where any node can send
messages to any other node. In general, handoff counters can work with less
connectivity. The assumptions that we make about network topology for the
remainder of the paper are:

\begin{itemize}
\item Each link is bidirectional.
\item The network restricted to tier 0 nodes is a connected sub-network.
\item For each node $u$, there is a path from $u$ to a tier 0 node along a
strictly descending chain of tiers.
\item If a tier node $u$ is linked to two smaller tier nodes $v$ and $w$, then
there is also a link between $v$ and $w$.
\end{itemize}

These assumptions allow version vector dissemination in tier 0 nodes, while
also allowing a client $u$ to start by exchanging messages with a server $v$
and later switching to a server $w$ if $v$ becomes unresponsive. These
assumptions are met by Example~\ref{ex:datacenter} and by the
Figure~\ref{fig:simpleTiers} topology, where inter-datacenter communication is
performed by tier 0 nodes, and where communication between tier 1 nodes or
between tier $n+1$ and tier $n$ nodes needs only be attempted within each
datacenter. It should be emphasized that these assumptions reflect only what
communications are attempted, and can be thought of as the rules for forming a
communication overlay; any of these links may be down for some time, even
possibly incurring in temporary network partitions.

\subsection{Distributed Algorithm}

A benefit of adopting the CRDT approach is not needing a complex distributed
algorithm -- the complexity is transferred to the CRDT. To achieve
correctness, basically any form of gossip can be used, where each node keeps
sending its state instance to randomly picked neighbors, and each node upon
receiving an instance merges it with the local one. To address efficiency,
concerns like the choice of neighbors to communicate will need attention; we
address such concerns in Section~\ref{practical_considerations}.

To describe both the mechanism and its correctness proofs we consider
Algorithm~\ref{alg:dist_algo} to be used, with operations defined in
Figures~\ref{HC-operations} and \ref{HC-transformations}.  In this algorithm,
each node $i$ has a local replica $C_i$, which is an instance of the Handoff
Counter CRDT.  A CRDT replica is initialized using the globally unique node id
and the node tier.  The local operations $\fetch_i$ and $\incr_i$ are
delegated to the corresponding CRDT ones. Each node periodically picks a
random neighbor $j$ and sends it the local instance. Upon receiving an
instance $C_j$ through a link $(j,i)$, it is merged with the local one,
through the CRDT merge operation.  The algorithm is quite trivial, all effort
being delegated to the Handoff Counter data type, namely through its $\init$,
$\fetch$, $\incr$ and $\merge$ operations.

Regarding fault tolerance, the state $C_i$ is assumed to be stored in stable
storage, and assigning to it is assumed to be an atomic operation. This means
that temporary variables used in computing a data-type operation do not need
to be in stable storage, and an operation can crash at any point before
completing, in which case $C_i$ will remain unchanged. The functional style
used in defining the CRDT and algorithm emphasizes this aspect.

\mathcode`|="326A          



\begin{algorithm}[t]
\DontPrintSemicolon
\SetKwBlock{constants}{constants:}{}
\SetKwBlock{state}{state:}{}
\SetKwBlock{periodically}{periodically}{}
\SetKwBlock{on}{on}{}

\constants{
  $i$, globally unique node id \;
  $t_i$, node $i$ tier \;
  $n_i$, set of neighbors \;
}
\BlankLine

\state{
  $C_i$, handoff counter data type; initially, $C_i = \init(i, t_i)$ \;
}
\BlankLine

\on({$\fetch_i$}){
  return $\fetch(C_i)$ \;
}
\BlankLine

\on({$\incr_i$}){
  $C_i := \incr(C_i)$ \;
}
\BlankLine

\on({$\receive_{j,i}(C_j)$}){
  $C_i := \merge(C_i, C_j)$ \;
}
\BlankLine

\periodically(){
  let $j = \random(n_i)$ \;
  $\send_{i,j}(C_i)$ \;
}
\BlankLine

\caption{Distributed algorithm for a generic node $i$.}
\label{alg:dist_algo}

\end{algorithm}


\subsection{Handoff Counter Data Type}

Unfortunately, the Handoff Counter data type is not so trivial. On one hand, a
server which is receiving counter values from a client should be able to
garbage collect client specific entries in its state; on the other hand,
neither duplicate or old messages should lead to over-counting, nor lost
messages lead to under-counting. Towards this, a handoff counter has state that
allows a 4-way handshake in which some accounted value (a number of
increments) is moved reliably from one node to the other. To understand the
essence of the mechanism, the steps when a node $i$ is handing-off
some value to a node $j$, when no messages are lost,
are the following, as exemplified in Figure~\ref{fig:Handoff}:

\begin{enumerate}
\item Node $i$ sends its $C_i$ to node $j$; node $j$ does $C'_j:=\merge(C_j,
C_i)$; the resulting $C'_j$ has a \emph{slot} created for node $i$;
\item Node $j$ sends $C'_j$ to $i$; node $i$ performs $C'_i:=\merge(C_i, C'_j)$; the
resulting $C'_i$ has a \emph{token} specifically created for that slot,
into which the locally accounted value has been moved;
\item Node $i$ sends $C'_i$ to node $j$; node $j$ does a $C''_j:=\merge(C'_j,
C'_i)$; this merge, seeing the token matching the slot, acquires the accounted
value in the token and removes the slot from the resulting $C''_j$;
\item Node $j$ sends $C''_j$ to $i$; node $i$ performs $C''_i:=\merge(C'_i, C''_j)$;
seeing that the slot is gone from $C''_j$, it removes the token from the
resulting $C''_i$.
\end{enumerate}

\begin{figure}
\centering
\includegraphics[width=0.8\textwidth]{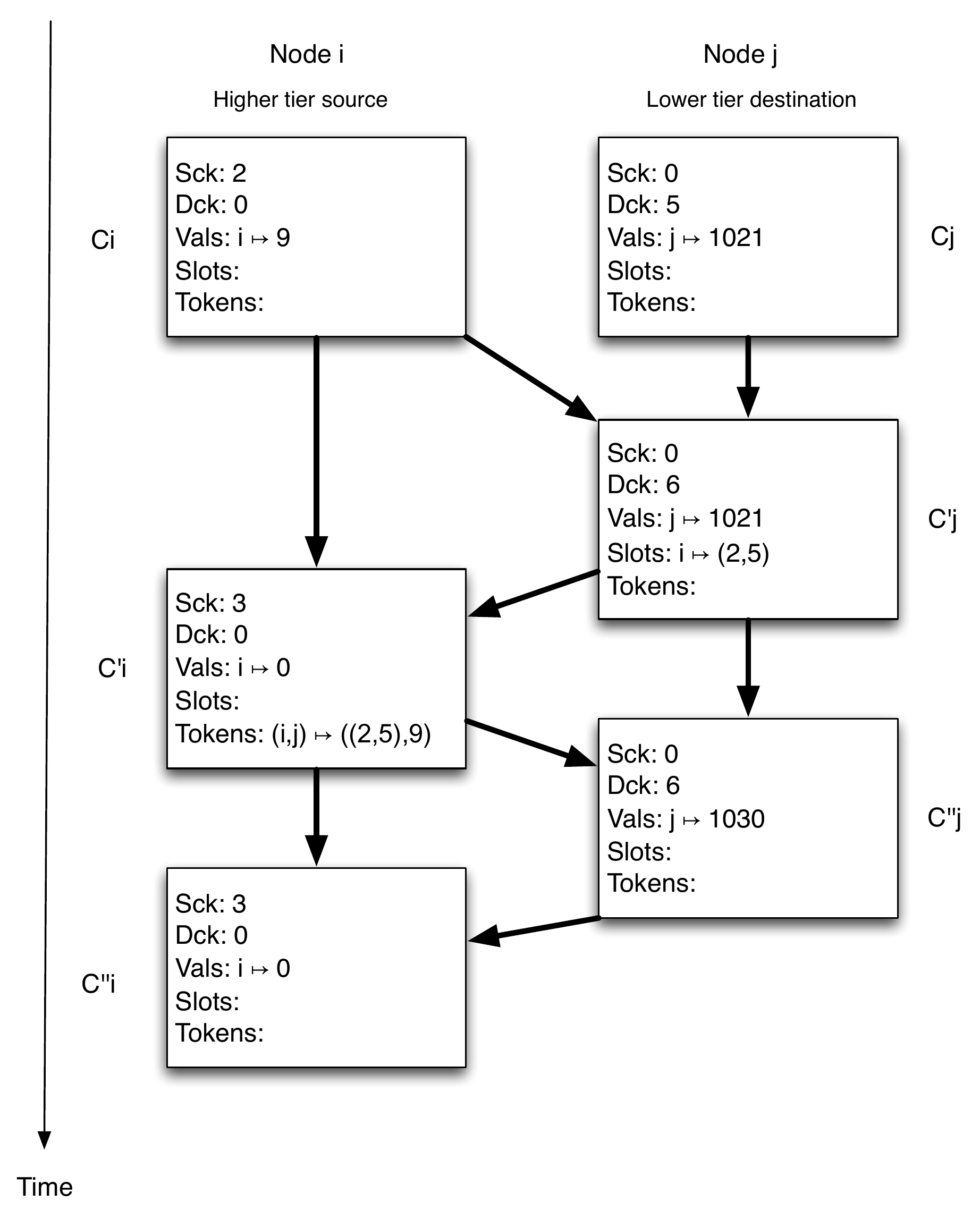}
\caption{A handoff from node i to j (only relevant fields are shown).}
\label{fig:Handoff}
\end{figure}

The end result of such an exchange is that an accounted value has been moved
from node $i$ to node $j$, and neither $i$ has a $j$ specific entry in its
state $C''_i$, nor $j$ has a $i$ specific entry in $C''_j$. Temporary entries
(slots and tokens) were garbage collected.

It should be noted that, although a given handoff takes 4 steps, when a pair
of nodes keep exchanging messages continuously, handoffs are pipelined: steps 1
and 3 are overlapped, as well as steps 2 and 4. When a message from a node $i$
arrives at a smaller tier node $j$, it typically carries a token, which is
acquired, filling the current slot at $j$, and a new slot is created; when the
``reply'' from $j$ arrives at $i$, it makes $i$ garbage collect the current
token, and a new token is created. This means that in the normal no-loss
scenario, each round-trip moves an accounted value (some number of increments)
from $i$ to $j$.

The mechanism must, however, ensure correctness no matter what communication
patterns may occur. Several properties are assured, namely:

\begin{itemize}
\item A given slot cannot be created more than once; even if it was created,
later removed and later a duplicate message arrives;
\item A token is created specifically for a given slot, and does not match
any other slot;
\item A given token cannot be created more than once; even if it was created,
later removed and later a duplicate message having the corresponding slot
arrives.
\end{itemize}

Towards this, the CRDT keeps a pair of logical clocks, \emph{source clock} and
\emph{destination clock}, that are used when a node plays the role of source
of handoff and destination of handoff, respectively.  A slot or token is
identified by the quadruple: source id, destination id, source clock,
destination clock.  When creating a slot the destination clock is incremented;
when creating a token for a given slot, the source clock is checked and
incremented.  This assures that neither a given slot nor a given token can be
created more than once.

Figure~\ref{fig:Handoff} shows an execution that moves a count of 9 from node
$i$ to node $j$, illustrating the evolution, after each merge, of the
subset of fields in the state that are more closely related to the handoff
procedure. 

The reason for two logical clocks per node (one would suffice for safety) is
to allow a middle tier node to play both roles (source and destination),
allowing the counter used for handoffs the node has started (as source) to
remain unchanged so that the handoff may complete, even if the node is a busy
server and is itself being the destination of handoffs. With a single clock
there would be the danger that during the round-trip to a smaller tier node,
messages from clients arrive, increasing the clock when creating a slot, and
making the handoff attempt to fail (as no token would be created) and need to
be repeated; under heavy load progress could be compromised.

Each handoff counter keeps a map $\vals$ with an entry which is only
incremented or added-to locally; in the case of tier 0 nodes there are also
entries regarding other tier 0 nodes. In other words, each node keeps a
structure similar to a version vector, mapping node ids to integers, which
has only the self entry in the case of non tier 0 nodes, and also other tier
0 nodes entries in the case of tier 0 nodes.

Towards ensuring local monotonicity in reporting, while allowing values to
move between nodes, each handoff counter has two integer fields,
$\val$ and $\below$, always updated in a non-decreasing way.
Field $\val$ keeps the maximum counter value that can be safely reported
according to local knowledge. Field $\below$ keeps a lower bound of values
accounted in strictly smaller tiers; namely, it summarizes knowledge about
tier 0 version vectors, avoiding the need for their dissemination to other
tiers.  The state of a handoff counter is then a record with the fields as
in Figure~\ref{HC-fields}.

\begin{figure}
\begin{description}
\item[$\id$:] node id;
\item[$\tier$:] node tier;
\item[$\val$:] maximum counter value that can be safely reported given local
knowledge;
\item[$\below$:] lower bound of values accounted in smaller tiers;
\item[$\vals$:] map from ids to integers; with a single self entry if tier other
than 0;
\item[$\sck$:] source clock -- logical clock incremented when creating tokens;
\item[$\dck$:] destination clock -- logical clock incremented when creating slots;
\item[$\slots$:] map from source ids to pairs $(sck, dck$) of logical clocks;
\item[$\tokens$:] map from pairs $(i, j)$ to pairs $((sck, dck), n)$
containing a pair of logical clocks and an integer;
\end{description}
\caption{Handoff Counter data type state (record fields)}
\label{HC-fields}
\end{figure}

The entries that keep tokens, slots and vals have been optimized, using maps as
opposed to sets, namely to obtain a fast lookup in the case of slots. This is
relevant as there may exist a considerable number of slots, depending on
the number of concurrent clients, while there are typically very few tokens
(just one in the more normal case). Such is possible due to the following:

\begin{itemize}
\item Each node $j$ needs to keep at most one slot for any given node $i$.
Therefore, a slot $(i, j, sck, dck)$ is kept as an entry mapping $i$ to
pairs $(sck, dck)$ in the slot map at $j$.
\item For each pair of nodes $i$ and $j$, there is only the need to keep
at most one token of the form $((i, j, sck, dck), n)$. However, such
token may be kept at nodes other than $i$. Tokens are stored in
maps from pairs $(i,j)$ to pairs  $((sck, dck), n)$.
\end{itemize}

The Handoff Counter data type is shown in Figure~\ref{HC-operations}.
The $\init$ operation creates a new CRDT instance (replica) for the counter;
it takes as parameters the node id and tier number;
it should be invoked only once for each globally unique id.
Operation $\fetch$ simply returns the $\val$ field, which caches the
higher counter value known so far.
Operation $\incr$ increments both the self entry in the ``version vector''
(i.e., $\vals_i(i)$) and the cached value in $\val$.
For the purposes of conciseness and clarity, in a definition of an operation
$\af{op}(C_i) \defeq \ldots$, the fields of $C_i$ can be accessed in the form
$\af{field}_i$, e.g., $\tokens_i$, and $C_i$ denotes a record with field $\id$
containing $i$; i.e., $C_i \defeq \{\id = i, \tier = \tier_i, \ldots \}$.

\begin{figure}
\begin{eqnarray*}
\init(i, t) &\defeq& \{\id = i, \tier = t,
 \val = 0, \below = 0, \sck = 0, \dck = 0,
\\ && \hphantom{\{}
 \slots=\{\}, \tokens = \{\}, \vals = \{i \mapsto 0\} \} \\
\fetch(C_i) &\defeq& \val_i \\
\incr(C_i) &\defeq& C_i\{\val = \val_i + 1, \vals = \vals_i\{i \mapsto
\vals_i(i) + 1\}\} \\
\merge(C_i, C_j) &\defeq&
\cachetokens(
\createtoken(
\discardtokens(
\aggregate(
\\&&
\mergevectors(
\createslot(
\discardslot(
\fillslots(
C_i, C_j), C_j),
\\&&
 C_j), C_j), C_j), C_j), C_j), C_j)
\end{eqnarray*}
\caption{Handoff Counter data type operations}
\label{HC-operations}
\end{figure}

The CRDT $\merge$ operation is by far the most complex one. It can be written
as the composition of 8 transformations that are successively applied, each
taking the result of the previous transformation and the received instance as
parameters. Each of these transformations takes care of a different aspect of
merging the two instances into a new one; they are presented in
Figure~\ref{HC-transformations}, where the following notation is used.

\paragraph{Notation}
We use mostly standard notation for sets and maps/relations. A map is a set of
$(k,v)$ pairs (a relation), where each $k$ is associated with a single $v$; to
emphasize the functional relationship we also use $k \mapsto v$ for entries in
a map. We use $M\{ \ldots \}$ for map update; $M\{x \mapsto 3\}$ maps $x$ to
$3$ and behaves like $M$ otherwise. For records we use similar notations but
with $=$ instead of $\mapsto$, to emphasize a fixed set of keys. We use
$\domainsub$ for domain subtraction; $S \domainsub M$ is the map obtained by
removing from $M$ all pairs $(k,v)$ with $k \in S$. We use set comprehension
of the forms $\{x \in S | P(x)\}$ or $\{f(x) | x \in S | P(x) \}$, or list
comprehensions, using square brackets instead.  The domain of a relation $R$
is denoted by $\dom(R)$, while $\fst(T)$ and $\snd(T)$ denote the first
and second component, respectively, of a tuple $T$.
 We use $\mergewith{f}{m}{m'}$ to represent joining maps while
applying $f$ to the values corresponding to common keys, i.e.,
$\mergewith{f}{m}{m'} =
\{(k,v) \in m | k \notin \dom(m')\} \union
\{(k,v) \in m' | k \notin \dom(m)\} \union
\{(k, f(m(k), m'(k))) | k \in \dom(m) \cap \dom(m')\}$.
To define a function or predicate by cases, we use $\kw{if} X \kw{then} Y
\kw{else} Z$ to mean ``$Y$ if $X$ is true, $Z$ otherwise''.

\begin{figure}
\begin{eqnarray*}
\fillslots(C_i, C_j) &\defeq&
C_i\{\vals = \vals_i\{i \mapsto \vals_i(i) + \sum[n | (\_,n) \in S]\},
\\&&
\hphantom{C_i\{}
 \slots = \dom(S) \domainsub \slots_i\}
\\&&
\kw{where} S \defeq \{ (src, n) | ((src, dst), (ck, n)) \in \tokens_j |
\\ &&
\hphantom{\kw{where} S \defeq \{ (src, n) | {}}
dst = i \land (src, ck) \in \slots_i \}
\\
\discardslot(C_i, C_j) &\defeq&
\kw{if} 
j \in \dom(\slots_i) \land \sck_j > \fst(\slots_i(j)) 
\\&&
\kw{then} C_i\{\slots = \{j\} \domainsub\slots_i\}
\\&&
\kw{else} C_i  \\
\createslot(C_i, C_j) &\defeq&
\kw{if}
\tier_i < \tier_j \land \vals_j(j) > 0 \land 
j \notin \dom(\slots_i)
\\&&
\kw{then}
C_i\{\slots = \slots_i\{j \mapsto (\sck_j, \dck_i)\},
\dck = \dck_i + 1\}
\\&&
\kw{else} C_i
\\
\mergevectors(C_i, C_j) &\defeq&
\begin{array}[t]{@{}l@{\ }l}
\kw{if} \tier_i = \tier_j = 0 &
\kw{then} C_i\{\vals = \mergewith{\max}{\vals_i}{\vals_j}\} \\
 & \kw{else} C_i
\end{array}
\\
\aggregate(C_i, C_j) &\defeq&
C_i\{\below = b, \val = v\}
\\&&
\kw{where} b \defeq
\begin{array}[t]{@{}l}
\kw{if} \tier_i = \tier_j
\kw{then} \max(\below_i, \below_j) \\
\kw{else}\ \kw{if}
\tier_i > \tier_j \kw{then}\max(\below_i, \val_j) \\
\kw{else} \below_i 
\end{array}
\\&&
\hphantom{\kw{where}{}} v \defeq
\begin{array}[t]{@{}l}
\kw{if} \tier_i = 0 \kw{then} \sum[n | (\_,n) \in \vals_i] \\
\kw{else}\ \kw{if}
\tier_i = \tier_j \kw{then} \max(\val_i, \val_j, b + \vals_i(i) + \vals_j(j)) \\
\kw{else} \max(\val_i, b + \vals_i(i))
\end{array}
\\
\discardtokens(C_i, C_j) &\defeq&
C_i\{ \tokens = \{ (k, v) \in \tokens_i | \neg P(k,v) \}
\\ &&
\kw{where}
P((src, dst), ((\_, dck), \_)) \defeq (dst = j) \land
\\ &&
\qquad\qquad
\begin{array}{l@{\ }l}
\kw{if} src \in \dom(\slots_j) & \kw{then} \snd(\slots_j(src)) > dck \\
                               & \kw{else} \dck_j > dck
\end{array}
\\
\createtoken(C_i, C_j) &\defeq&
\kw{if} i \in \dom(\slots_j) \land \fst(\slots_j(i)) = \sck_i
\\&&
\kw{then}
\begin{array}[t]{@{}l@{}l}
C_i\{ & \tokens = \tokens_i\{(i,j) \mapsto (\slots_j(i), \vals_i(i))\}, \\
& \vals = \vals_i\{i \mapsto 0\}, \\
& \sck = \sck_i + 1\}
\end{array}
\\&&
\kw{else} C_i
\\
\cachetokens(C_i, C_j) &\defeq&
\kw{if} \tier_i < \tier_j
\\&&
\kw{then} C_i\{\tokens = \mergewith{f}{\tokens_i}{t}\}
\\&&
\hphantom{\kw{then}}\kw{where} 
t \defeq \{((src, dst), v) \in \tokens_j | src = j \land dst \neq i\},
\\ &&
\hphantom{\kw{then}\kw{where}{}}
f((ck, n), (ck', n')) \defeq
\begin{array}[t]{l@{\ }l}
\kw{if} \fst(ck) \geq \fst(ck') &
\kw{then} (ck, n) \\
& \kw{else} (ck', n')
\end{array}
\\&&
\kw{else} C_i
\end{eqnarray*}
\caption{Handoff Counter auxiliary transformations in $\merge$.}
\label{HC-transformations}
\end{figure}

\paragraph{}
We now describe informally each of these transformations used in $\merge$.
We leave the formal proof of correctness for the next section.

\paragraph{$\fillslots(C_i, C_j)$}
Fills all possible slots in $C_i$ for which there are matching tokens in
$C_j$, removing them from $\slots_i$ and adding the counter values in the
tokens to the self version vector entry $\vals_i(i)$. We call this transfer of
a counter value in a token to the $\vals$ entry of the destination node
\emph{acquiring the token} and the corresponding slot removal \emph{filling
the slot}.

\paragraph{$\discardslot(C_i, C_j)$}
Discards a slot, if any, in $C_i$ for source $j$, that cannot ever be
possibly filled by a matching token, because $C_j$ assures no such token
will ever be generated. For such a slot that still remains in $C_i$ (and,
therefore, has not been filled by a matching token in $C_j$ by the
just applied $\fillslots$) this is the case if the source clock at $C_j$ is
greater than the corresponding value in the slot.

\paragraph{$\createslot(C_i, C_j)$}
Creates a slot in $C_i$ for a higher tier source node $j$, if $C_j$ has some
non-zero value to handoff and there is no slot for $j$ at $C_i$.
If a slot is created, the local destination clock is stored in it and
increased, preventing duplicate creation of the same slot in case of duplicate
messages.
A slot merely opens the possibility of node $j$ creating a
corresponding token; not all slots will have corresponding tokens, some will
be discarded or replaced by newer slots. In fact, the local knowledge in $C_i$
after discarding $j$ specific information makes it impossible to avoid creating
a slot upon a slow or duplicate message; this is not a problem
as such slots will never match any token, being eventually discarded upon
further communication.

\paragraph{$\mergevectors(C_i, C_j)$}
Merges the corresponding version vector entries, doing a pointwise maximum for
common ids. This is only done when merging two tier 0 replicas.

\paragraph{$\aggregate(C_i, C_j)$}
Performs a vertical aggregation step (from smaller to higher tiers up to the
current one), that updates the $\below$ and $\val$ fields according to the
knowledge provided by $C_j$. This can never decrease their current values.
The effect is to propagate values accounted in tier 0 version
vectors, while adding knowledge provided by intermediate nodes, up to $C_i$
and $C_j$.

\paragraph{$\discardtokens(C_i, C_j)$}
Discards from $C_i$ the tokens that have already been acquired by $C_j$;
this is the case for tokens with id $(src, j, \_, dck)$ if either there is
a slot $(src, j, \_, dck')$ at $C_j$ with $dck' > dck$ or if there is no slot
$(src, j, \_, \_)$ at $C_j$ and $\dck_j > dck$. Here $src$ is normally, but
not necessarily, equal to $i$, as $C_i$ can cache tokens from another source than $i$.

\paragraph{$\createtoken(C_i, C_j)$}
Creates a token, to which the currently accounted value in $\vals_i(i)$ is
moved, if there is a slot for $i$ in $C_j$ having a source clock equal to the
current one at $C_i$.
If a token is created, the local source clock is
increased, preventing duplicate creation of the same token in case of
duplicate messages.

\paragraph{$\cachetokens(C_i, C_j)$}
Keeps a copy of tokens generated at a higher tier node $j$ meant to some other
destination $k$. For each pair source-destination, older tokens (that must
have already been acquired) are replaced by newer ones. Caching tokens
provides availability under faults, as it allows a client $j$ in the middle of
a handoff to $k$, to delegate to $i$ the responsibility of finishing the
handoff, in case $j$ wants to terminate but either $k$ has crashed and is
recovering or the link between $j$ and $k$ is currently down.  Only tokens
that have been generated at node $j$ are cached (other tokens currently cached
at $j$ are not considered) so that alternate handoff routes are provided,
while preventing the flooding and the large sets of tokens that would result
from a transitive dissemination of tokens to other nodes.

\subsection{Implementation and Experimental Validation}

The development of Handoff Counters was made in tandem with a prototype
implementation and checked by a testing infrastructure that allows for
randomized runs over given topologies. These runs exercised the solution
robustness, in particular by replaying old messages out of order, and allowed
to detect and correct several corner cases, enabling the correction of subtle
bugs that existed in tentative versions of the mechanism.

Although testing does not ensure correctness, the current implementation has
successfully passed randomized traces with up to one hundred million steps,
both in giving the correct result and also in garbage collecting temporary
entries, making it a complement to a manual formal proof (which is prone to
human errors).

An experimentally robust solution was thus a prelude to the formal
correctness proof in the following section, and added an independent
assessment to the overall approach. The implementation and testing
infrastructure, written in Clojure, is publicly available in GitHub
(\url{https://github.com/pssalmeida/clj-crdt}), and serves as a complement to
this article.

\section{Correctness}
\label{correctness}


\begin{lemma} \label{lem-slot-once}
Any slot $(s, d, sck, dck)$ can be created at most once.
\end{lemma}
\begin{proof}
Each node uses its own id as $d$ in the slot; therefore, no two nodes can
create the same slot. In each node, a slot is only created when applying
$\createslot$, which also increments $\dck$ upon storing it in the slot;
therefore, a subsequent $\createslot$ in the same node cannot create the same
slot.
\end{proof}

\begin{lemma}
Any token with id $(s, d, sck, dck)$ can be created at most once.
\end{lemma}
\begin{proof}
Each node uses its own id as $s$ when creating the token; therefore, no two
nodes can create the same token. In each node $\createtoken$ also increments
$\sck$ upon storing it in the token; therefore, a subsequent $\createtoken$ in
the same node cannot create the same token.
\end{proof}

\begin{lemma}
Any token with id $(s, d, sck, dck)$ can be acquired at most once.
\end{lemma}
\begin{proof}
Such token can only be acquired in node $d$ having a corresponding slot 
$(s, d, sck, dck)$ while performing the $\fillslots$ function, which
removes this slot from the resulting state, preventing a subsequent acquisition
of the same token, as due to Lemma~\ref{lem-slot-once} this slot cannot be
recreated.
\end{proof}

\begin{proposition} \label{prop-token-created}
Given a token $T$ with id $(s, d, sck, dck)$:
(i) $T$ will not be removed from any node before it has been acquired;
(ii) a corresponding slot $S$ will exist in node $d$ between the time when
$T$ is created and the time when $T$ is acquired.
\end{proposition}
\begin{proof}
By induction on the trace of the actions performed by the system, using
(i) and (ii) together in the induction hypothesis. The only relevant action is
$\receive_{j,i}(C_j)$ of a message previously sent by a node $j$ to a node
$i$, with the corresponding merge being applied to the state of $i$. Given
the asynchronous system model allowing message duplicates, we must assume that
$C_j$ can be the state of $j$ at any point in the past, due to some
$\send_{j,i}$ action that resulted in a message arbitrarily delayed and/or
duplicated.

Regarding (i), a token $T$ with id $(s, d, sck, dck)$ can only be removed
either: (1) In $\discardtokens$ when merging with a state $C$ from $d$, having
either a slot $(s, d, \_, dck')$ with $dck' > dck$, or with no slot for $s$
and the destination clock in $C$ greater than $dck$; either way, given that
slots for destination $d$ are created with increasing destination clock
values, it implies that $C$ corresponds to a time after slot $S$ was created.
By the induction hypothesis, $S$ would have existed until $T$ was acquired,
and as $S$ is absent from $C$, this implies that $T$ was already acquired. (2) Or
$T$ could be removed by the $\tokens$ map entry for $(s,d)$ being overwritten
in $\createtoken$; this last case is not possible because: tokens are created
using the current source clock, which is then incremented; for a token with id
$(s, d, \sck_s, \_)$ to be created, a received counter state $C$ from $d$ must
contain a slot $(s, d, sck', \_)$ with $sck' = \sck_s$. This means that $d$
would have previously received $T$ from $s$ and from the induction hypothesis,
$d$ would have had a corresponding slot and would have already acquired $T$
filling the slot. When $C$ arrived at $s$, $T$ would be discarded by
$\discardtokens$ before invoking $\createtoken$, i.e., as in the first case
above.

Regarding (ii), $T$ is created, in $\createtoken$, only if a corresponding slot
has been created at some previous time in node $d$; this slot can only be
removed either: in $\fillslots$, when $T$ is acquired; or in $\discardslot$,
when merging with a counter $C$ from node $s$ whose source clock is greater
than $sck$, implying a state in $s$ after $T$ has been created. But in this
case this slot at node $d$ cannot reach the $\discardslot$ function
as, by the induction hypothesis, $T$ would be present in $C$ and would have
been acquired by $\fillslots$, filling the slot, just before invoking
$\discardslot$.
\end{proof}

\begin{lemma}
\label{token-eventually-acquired}
Any token with id $(s, d, sck, dck)$ will be eventually acquired.
\end{lemma}
\begin{proof}
Such token, created at node $s$, must have resulted from direct message
exchanges between $s$ and a smaller tier node $d$.
From Proposition~\ref{prop-token-created}, this token will remain at $s$
and a corresponding slot will exist at $d$ until the token is acquired.
As $s$ will keep sending its $C_s$ to its neighbors, therefore to $d$, and
from the system model assumptions (Section~\ref{system-model}) messages
eventually get through, if the token has not yet been acquired (e.g., by
communication between $d$ and some other node caching the token),
$C_s$ containing the token will eventually arrive at $d$, which will
acquire it.
\end{proof}

\begin{definition}[Enabled token]
A token $((s, d, sck, dck), n)$ is called \emph{enabled} if there exists a
corresponding slot $(s, d, sck, dck)$ at node $d$. The set of enabled tokens
in a configuration (set of replicas) $C$ is denoted $E_C$.
\end{definition}

Whether a token is enabled is a global property, not decidable by nodes
holding the token with local information.  From
Proposition~\ref{prop-token-created}, a token $T$ is enabled when created; it
remains enabled until it is acquired, when the corresponding slot is filled.
We also remark that $E_C$ is the \emph{set} (as opposed to multiset) of enabled
tokens; the presence of duplicates of some token $T$ created at one node and
cached at some other node(s) is irrelevant.

\begin{lemma}
\label{below-val-non-decreasing}
Fields $\below$ and $\val$ are non-decreasing.
\end{lemma}
\begin{proof}
By induction on the trace of the actions performed by the system. These fields
only change by an $\incr_i$ at node $i$, which increments $\val_i$, or when
doing a merge when receiving a message, in $\aggregate$, which either updates
$\below$ and $\val$ using a maximum involving the respective current value,
or stores in $\val_i$ the sum of $\vals_i$ entries, if $i$ is a tier 0 node,
which is also non-decreasing, as $\vals_i$ for tier 0 nodes contain a set of
entries always from tier 0 nodes, only updated by a pointwise maximum (as tier
0 nodes never create tokens).
\end{proof}

\auxfun{CTV}

\begin{definition}[Cumulative Tier Value]
In a configuration $C$, the cumulative tier value for tier $k$, written
$\CTV_C(k)$, is the sum, for all nodes with tier up to $k$,
of the self component of the $\vals$ field plus the tokens created by
these nodes that are still enabled, i.e.:
\begin{eqnarray*}
\CTV_C(k) &\defeq &\sum [ \vals_i(i) | C_i \in C | \tier_i \leq k ] + {} \\
                 &&\sum [ n | ((i,\_,\_,\_),n) \in E_C | \tier_i \leq k ].
\end{eqnarray*}
\end{definition}

\begin{lemma}
\label{CTV-non-decreasing}
For each $k$, $\CTV_C(k)$ is non-decreasing, i.e., for any
transition between configurations $C$ and $C'$, $\CTV_C(k) \leq \CTV_{C'}(k)$.
\end{lemma}
\begin{proof}
For any node $i$, the only time $\vals_i(i)$ can decrease is when a token is
created, enabled, and the value is moved to the token; in this case
$\CTV_{C'}(k)$ remains unchanged for all $k$. When a token holding value $n$
ceases to be enabled (being acquired), $n$ is added to the $\vals_j(j)$ field
for some smaller tier node $j$; this makes $\CTV_{C'}(k)$ either unchanged or
greater.
\end{proof}

\begin{lemma}
\label{CTV-monotonic}
$\CTV_C(k)$ is monotonic over $k$, i.e.,
$k_1 \leq k_2 \implies \CTV_C(k_1) \leq \CTV_C(k_2)$.
\end{lemma}
\begin{proof}
Trivial from the $\CTV$ definition.
\end{proof}

\begin{proposition}
\label{prop-below-val}
For any counter replica $C_i$ in a configuration $C$:
(i) $\below_i \leq \CTV_C(\tier_i - 1)$;
(ii) $\val_i \leq \CTV_C(\tier_i)$.
\end{proposition}
\begin{proof}
By induction on the trace of the actions performed by the system, using
(i) and (ii) together in the induction hypothesis.
Given that $\CTV_C(k)$ is non-decreasing (by Lemma~\ref{CTV-non-decreasing}),
so are the right-hand sides of the inequalities, and the only relevant actions
are those that update either $\below_i$ or $\val_i$:
(1) An increment $\incr_i$ at node $i$, resulting in an increment of both
$\val_i$ and $\vals_i(i)$, in which case the inequality remains true. 
(2) A $\receive_{j,i}(M)$ of a message previously sent by a node $j$ to node
$i$, with the corresponding merge being applied to the state of $i$, and the
fields being updated by $\aggregate$. Regarding $\below_i$, there are three
cases: it remains unchanged, it can be possibly set to $\below_j$ if $\tier_j
= \tier_i$, or it can be possibly set to $\val_j$ if $\tier_j < \tier_i$; in
each case the induction hypothesis is preserved because $\CTV_C$ is
non-decreasing ($M$ can be any message from node $j$, arbitrarily from the
past) and in the last case also due to the monotonicity of $\CTV_C(k)$ over
$k$ (by Lemma~\ref{CTV-monotonic}).
Regarding $\val_i$, either it is set to the sum of the $\vals_i$ values, if
$\tier_i = 0$, which does not exceed $\CTV_C(0)$ due to the pointwise maximum
updating of $\vals$ fields for tier 0 nodes; or it either remains unchanged,
is set to $\val_j$ only if $\tier_i = \tier_j$, or is set to the sum of the
values (computed for the next configuration) of $\below_i$ with $\vals_i(i)$
and also $\vals_j(j)$ when $\tier_i = \tier_j$; in each case the induction
hypothesis is preserved.
\end{proof}

\begin{proposition}
\label{prop-global-increments}
The number of increments globally issued up to any time $t$, say $I^t$, is
equal to the sum of the values held in the set of enabled tokens and of the
self entries in the $\vals$ field of all nodes; i.e., for a network having
maximum tier $T$, given a configuration $C^t$ at time $t$, we have
$I^t = \CTV_{C^t}(T)$.
\end{proposition}
\begin{proof}
By induction on the trace of the actions performed by the system. The relevant
actions are an increment at some node $i$, which results in an increment of
the $i$ component of the $\vals$ field of node $i$; or a
$\receive_{j,i}(C_j)$ of a message previously sent by a node $j$ to a node
$i$, with the corresponding merge being applied to the state of $i$, leading
possibly to:
the filling of one or more slots, each slot $S$ corresponding to a token $(S,
n)$, which adds $n$ to $\vals_i(i)$ and removes slot $S$ from $i$, which makes
the token no longer enabled, leaving the sum unchanged;
discarding a slot, which cannot, however, correspond to an enabled token, as
from Proposition~\ref{prop-token-created} a slot will exist until the
corresponding token is acquired;
merging $\vals$ pointwise for two tier 0 nodes, which does not change
the self component $vals_i(i)$ of any node $i$;
discarding tokens, which cannot be enabled because, by
Proposition~\ref{prop-token-created}, tokens are only removed from any node
after being acquired;
the creation of an enabled token $(\_, n)$, at node $i$, holding
the value $n = \vals_i(i)$ and resetting $\vals_i(i)$ to 0, leaving the sum unchanged;
caching an existing token, which does not change the set of tokens in the
system.
\end{proof}

\begin{proposition} \label{prop-ECDC-bi}
Any execution of Handoff Counters ensures ECDC fetch bounded by increments.
\end{proposition}
\begin{proof}
In any configuration $C$, a $\fetch_i$ at replica $C_i$ simply returns $\val_i$.
From Proposition~\ref{prop-below-val}, this value does not exceed
$\CTV_C(\tier_i)$, which, from the monotonicity of $\CTV_C(k)$
(Lemma~\ref{CTV-monotonic}) and Proposition~\ref{prop-global-increments}, does
not exceed the number of globally issued increments.
\end{proof}

\begin{proposition} \label{prop-ECDC-lm}
Any execution of Handoff Counters ensures ECDC local monotonicity.
\end{proposition}
\begin{proof}
Operation $\fetch_i$ simply returns $\val_i$, which is non-decreasing
(Lemma~\ref{below-val-non-decreasing}) and which is always incremented
upon a local $incr_i$; therefore, for any node $i$ the difference between
$\fetch_i$ at two points in time will be at least the number of increments
issued at node $i$ in that time interval.
\end{proof}

\begin{proposition} \label{prop-ECDC-ea}
Any execution of Handoff Counters ensures ECDC eventual accounting.
\end{proposition}
\begin{proof}
Let $T$ be the maximum node tier in the network, and $N$ the set of nodes.
From Proposition~\ref{prop-global-increments}, the number of increments $I^t$
globally issued up to any time $t$,  for a configuration $C^t$, is equal to
$\CTV_{C^t}(T)$.
From Lemma~\ref{token-eventually-acquired}, by some later time $t' > t$,
all tokens from tier $T$ enabled at time $t$ will have been acquired by
smaller tier nodes (if $T > 0$ ), and also because $\mathsf{CTV}$ is non-decreasing, it
follows that
$I^t \leq \CTV_{C^{t'}}(T')$, for some $T' < T$. Repeating this reasoning
along a finite chain $T > T' > \cdots > 0$, by some later time $t''$ we have
$I^t \leq \CTV_{C^{t''}}(0) = \sum[\vals_i^{t''}(i) | C_i^{t''} \in C^{t''} | \tier_i = 0]$, this last
equality holds because there are no tokens created at tier 0; given the network
topology assumptions of tier 0 connectedness, eventually at some later
time $t'''$, all $\vals$ entries in all tier 0 nodes will be pointwise greater
than the corresponding self entry at time $t''$, i.e., $\vals_i^{t'''}(j) \geq
\vals_j^{t''}(j)$ for all tier 0 nodes $i$ and $j$.
Given the topology assumptions of the existence, for each node $i$, of a path
along a strictly descending chain of tiers $\tier_i > \cdots > 0$, eventually,
by the $\aggregate$ that is performed when merging a received counter,
repeated along the reverse of this path, at some later time the $\val_i$ field
for each node $i$ in the network will have a value not less than the sum
above, and therefore, not less than the number of increments $I^t$
globally issued up to time $t$, which will be returned in the $\fetch_i$
operation.
\end{proof}

\begin{theorem}
Handoff Counters implement Eventually Consistent Distributed Counters.
\end{theorem}
\begin{proof}
Combine Propositions \ref{prop-ECDC-bi}, \ref{prop-ECDC-lm}, and
\ref{prop-ECDC-ea}.
\end{proof}

%

\section{Practical Considerations and Enhancements}
\label{practical_considerations}

Handoff Counters were presented as a general CRDT, that works under a simple
gossip control algorithm ``send counter to all neighbors and merge received
counters''. Here we discuss practical issues and outline some enhancements such
as more selective communication, how to amortize the cost of durable writes while
ensuring correctness, how to avoid sending the full CRDT state for busy servers
with many concurrent clients, and the issue of client retirement. A formal
treatment of these issues is deferred to further work. 

\subsection{Topology and Message Exchanges}

We have described a mechanism which is arbitrarily scalable, through the use
of any suitable number of tiers. In Example~\ref{ex:datacenter} we have
described a 3 tier scenario (tiers 0 for permanent nodes, tier 1 for serving
nodes, appropriate to the number of end-clients, and tier 2 for end-clients).

In practice, the most common deployment will probably consist of only tier 0
and 1 nodes (tier 0 for the data-store and tier 1 for the end-clients). This
is because a high scale scenario typically involves not only many clients, but
also many counters, with requests spread over those many counters. Having a
couple of tier 0 nodes per data-center per counter will cover most common
usages.

For presentation purposes, the distributed algorithm consisted simply of a
general gossip, where each node keeps sending its counter CRDT to each
neighbor. In practice, in the role of client, a node will simply choose one
smaller tier neighbor as server, to use in the message exchange, to maximize the
effectiveness of the handoff. Not only this avoids extra token caching by
other nodes and subsequent extra work in removing them after they have been
acquired, but also avoids the possible creation of slots that will have no
chance of being filled and will have to be discarded.
Only when a client suspects that the chosen server is down or there is a
network partition should another node be chosen as server, to continue the
exchange.

\subsection{Fault Tolerance}

The mechanism correctness assumes durable storage, and that the CRDT resulting
from each operation is successfully stored durably. We leave it as orthogonal
to our mechanism the way each node achieves local durability and local node
fault tolerance (e.g., through the use of storage redundancy, like using a
RAID, or by running a consensus protocol over a small number of tightly
connected machines emulating a single node). 

A practical issue that arises in each replica is that, to avoid unbearable
loss of performance, the actual write to durable storage (e.g., using POSIX's
fsync) should not be made after each operation over the CRDT. But if the write
to durable storage is delayed and messages continue to be exchanged, a node
crash will violate the correctness assumptions, as the local in-memory CRDT
value which could already have been received and merged by other nodes will be
lost.

To overcome this problem, a maximum frequency of durable writes can be defined.
Between durable writes, all CRDTs received from other nodes can be merged to
the transient in-memory CRDT, but no ``replies'' are sent; instead, their node
ids are collected in a transient set. After the durable write, messages
containing the written CRDT are sent to those nodes in the collected set of
ids, while new messages received are applied to the transient in-memory CRDT
and node ids are again collected into a new set, until the next durable write.

This means that all messages sent correspond to a durably stored CRDT; if the
node crashes the transient state (CRDT and set of node ids) is lost,
but this is equivalent to the messages received since the last durable write
having been lost. As the mechanism supports arbitrary message loss,
correctness will not be compromised.
This solution amortizes the cost of a durable write over many received requests,
essential for a heavily loaded server. Under little load, durable writes can
be made immediately and a reply message sent.

The maximum frequency of writes can be tuned (e.g., using some value between
100 and 1000 times per second) according to both storage device
characteristics and network latency. As an example, if clients of a given node
are spread geographically and there is considerable latency (e.g., 50ms),
waiting some time (e.g., 5 ms) to collect and merge messages from several
clients before writing to durable storage and replying should not cause a
noticeable impact.

\subsection{Restricting Transmitted State through Views}

\auxfun{view}

The algorithm as described adopts the standard CRDT philosophy, in which the
full CRDT state is sent in a message to be merged at the destination node.
For Handoff Counters we can explore the way merge works, to avoid sending the
full CRDT state, and instead only sending the state which is relevant for the
destination node. This strategy assumes that messages are sent to specific
nodes (as opposed to, e.g., being broadcast) and that the sender knows the
node id and tier of the message destination.  It also assumes that server
nodes that a given client uses for handoff are all of the same tier.  This
assumption is reasonable, and met by the examples discussed, where clients of
tier $n+1$ handoff to nodes of tier $n$.

The insight is that when a node $i$ is sending $C_i$ to a greater tier node
$j$, in what regards the $\slots$ field, only the entry for node $j$
is relevant when the merge is applied at $j$; the other entries are ignored
when $\merge(C_j, C_i)$ is performed at $j$, and can be omitted from the CRDT
to be sent. When $i$ is sending to a smaller tier node $j$, no slot from $C_i$
is relevant for the merge at $j$ and so, no slots need to be sent. Only when
communicating with a node $j$ of the same tier must the full $\slots_i$ map be
sent, as $j$ may be caching tokens from some greater tier client, with $i$ as
destination.

Using the insight above, instead of doing a $\send_{i,j}(C_i)$, node $i$ can
make use of a function $\view$ to restrict the state to the information
relevant to node $j$, and do a $\send_{i,j}(\view(C_i, j))$. This function can
be defined as:

\begin{eqnarray*}
\view(C_i, j) &\defeq&
  \kw{if} \tier_i < \tier_j \kw{then}
      C_i\{\slots = \{(k,s) \in \slots_i | k = j \} \} \\
&& \kw{else}\ \kw{if} \tier_i > \tier_j \kw{then}
      C_i\{\slots = \{\} \} \\
&& \kw{else} C_i.
\end{eqnarray*}

Even though this only involves the $\slots$ field, this component will
constitute the largest part of the counter state in a busy server with many
concurrent clients, as it can have one slot per client. This optimization will
allow sending only a small message to each client, and also avoid sending
slots to smaller tier nodes (e.g., when a tier 1 node communicates with tier
0).

\subsection{Client Retirement}

Given that each node accounts locally issued increments until they are handed
off, when an end-client has stopped issuing increments and wants to retire,
it should continue exchanging messages until it is certain that those
increments will be accounted elsewhere (in smaller tier nodes).

The normal way of doing so is to keep exchanging messages with the chosen
server, until the $\vals$ self component is zero and the $\tokens$ map is
empty. This can, however, mean waiting until a partition heals, if there is a
token for a partitioned server.  The token caching mechanism allows the client
to start a message exchange with an alternate server, which will cache the
token, to be delivered later to the destination.

While an end-client $i$ wishes to remain active, even if some node $k$ has
already cached a token from $i$ to server $j$, client $i$ cannot discard the
token unless it communicates with $j$ after $j$ has acquired it; otherwise, it
could cause an incorrect slot discarding at $j$.
But in the retirement scenario, if $\vals_i(i)$ self component is zero, and $i$
has learned that all its tokens are already cached at other nodes (by having
seen them in messages received from those nodes), $i$ can stop sending
messages and retire definitely. As no more messages are sent, no incorrect
slot removal will occur, and as all tokens from $i$ are cached elsewhere they
will be eventually acquired, implying a correct eventual accounting of all
increments issued at $i$.


Another issue regarding client retirement is slot garbage collection. The
mechanism was designed to always ensure correctness, and to allow temporary
entries (slots and tokens) to be removed in typical runs. As such, slots must
be kept until there is no possibility of them being filled. The mechanism was
designed so that a server can remain partitioned an arbitrary amount of time
after which a message arrives containing a token. This raises the possibility
that: a client $C$ sends a message to a server $S_1$, a slot is created at
$S_1$, a partition occurs just before a corresponding token is created at $C$,
the client starts exchanging messages to another server $S_2$ and successfully
hands off the local value to $S_2$ and retires; in this scenario, the slot at
$S_1$ will never be garbage collected, as $C$ is no longer alive to
communicate with $S_1$. (Under our system model $C$ is not expected to retire
for ever, and all partitions eventually heal, but dealing with client
retirement is a relevant practical extension.)

In this example, even though correctness was not compromised, each such
occurrence will lead to irreversible state increase which, even if
incomparable in magnitude to the scenario of naive CRDTs with client ids
always polluting the state, is nevertheless undesirable. This motivates a
complementary mechanism to improve slot garbage collection: if a client starts
using more than one server, it keeps the set of server ids used; when it
wishes to retire the intention is communicated to the server, together with
the set of server ids, until the retirement is complete; a server which
receives such intention keeps a copy of the last token by that client (in a
separate data-structure, independently of whether the server caches or acquires
the token), and starts an algorithm which disseminates the token to the set of
servers used by the client and removes it after all have acknowledged the
receipt. The insight is that when one of these servers sees the token, it can
remove any slot for that client with an older source clock. For this, it is
essential that this information is piggy-backed in the normal messages between
servers carrying the CRDT, and processed after the normal merge, so that a
server that has a slot corresponding to an enabled token for that client, that
may be cached in another server will see the token and fill the corresponding
slot, before attempting slot garbage collection by this complementary
mechanism.

\section{Beyond Counters}
\label{beyond-counters}

We have up to now addressed distributed counters, given their wide
applicability and importance. Using counters was also useful for presentation
purposes, as something concrete and widely known. The resulting mechanism and
lessons learned are, however, applicable far beyond simple counters.

What we have devised is a mechanism which allows some value to be handed off
reliably over unreliable networks, through multiple paths to allow
availability in the face of temporary node failures or network partitions.
Values are moved from one place to another by ``zeroing''  the origin and
later ``adding'' to the destination. Reporting is made by aggregating in two
dimensions: ``adding'' values and taking the ``maximum'' of values.
The value accounted at each node is updated by a commutative and associative
operation which ``inflates'' the value.
This prompts a generalization from simple counters over non-negative integers
to more general domains.

\newcommand\zero{\ensuremath{\mathbf{0}}}

The handoff counter CRDT can be generalized to any commutative monoid $M$ (an
algebraic structure with an associative and commutative binary operation
($\oplus$) and an identity element ($\zero$)) which is also a join-semilattice
(a set with a partial order ($\pleq$) for which there is a least upper bound
($x \join y$) for any two elements $x$ and $y$) with a least element ($\bot$), as long as it
also satisfies:

\begin{eqnarray*}
\bot &=& \zero \\
x \join y &\pleq& x \oplus y
\end{eqnarray*}

The CRDT state and definition of merge remain unchanged, except:
\begin{itemize}
\item
Fields $\val$, $\below$, range of $\vals$ entries and token payload now store
elements of $M$ instead of simple integers;
\item
Those fields are initialized to $\zero$ in the initialization of the CRDT;
$\zero$ is also used for resetting the self $\vals_i(i)$ entry in
$\createtoken$;
\item
The sum operation ($+$) over elements of the above fields is replaced by the
$\oplus$ operation;
\item
The $\max$ operation used in $\mergevectors$ and $\aggregate$ is replaced by
the $\join$ operation;
\end{itemize}

In terms of client-visible mutation operations, instead of $\incr$, any set of
operations that are associative and commutative and that can be described as
inflations over elements of $M$ (i.e., such that $x \pleq f(x)$) can be made
available.

In terms of reporting operations, instead of $\fetch$, the data type can make
available functions that can be defined over elements of $M$ (that result
from the aggregation made resorting to $\oplus$ and $\join$).

\subsection{Example: Map of Counters}

Sometimes more than a single counter is needed. Instead of having a group of
Handoff Counter CRDTs, a new CRDT can be devised, that holds a group of
counters together, made available as a map from counter id to value.
This will allow amortizing the cost of the CRDT state over the group of
counters, instead of having per-counter overhead.

The CRDT for the map-of-counters can then be defined by making elements of $M$
be maps from ids to integers and defining:
\begin{eqnarray*}
\zero &\defeq& \{\} \\
x \oplus y &\defeq& \union^+(x,y) \\
x \join y &\defeq& \union^{\max}(x,y)
\end{eqnarray*}

and by making available:

\begin{eqnarray*}
\fetch(C_i, c) &\defeq& \val_i(c) \\
\incr(C_i, c) &\defeq& C_i\{\val = \union^+(\val_i, \{c \mapsto 1\}),
                            \vals = \vals_i\{i \mapsto \union^+(\vals_i(i), \{c \mapsto 1\})\}\}
\end{eqnarray*}

\subsection{Example: PN-Counter}

A PN-Counter~\cite{rep:syn:sh138} can be both incremented and decremented. It
can be implemented as a special case of the previous example, with two entries
in the map: the $\af{p}$ entry, counting increments, and the $\af{n}$ entry,
counting decrements. The fetch operation returns the difference between these
values:

\auxfun{decr}

\begin{eqnarray*}
\fetch(C_i) &\defeq& \fetch(C_i, \af{p}) - \fetch(C_i, \af{n}) \\
\incr(C_i) &\defeq& \incr(C_i, \af{p}) \\
\decr(C_i) &\defeq& \incr(C_i, \af{n})
\end{eqnarray*}

\section{Discussion}

The standard approach to achieving reliability and availability in distributed
systems is to use a replicated service and distributed transactions, with a
fault tolerant distributed commit protocol, that works if some majority of
nodes are working (and not partitioned), e.g., Paxos
Commit~\cite{GrayLamport2006}.
This standard approach attacks several problems in the same framework:
network failures, temporary node failures and permanent node failures.
By doing so, it incurs a performance cost, due to the need to communicate
with several nodes, even when no failures occur.
Our approach does not impose such cost: when no failures occur, a node playing
the role of client only communicates with just one server.

Regarding availability, our approach (even after assuring that increments were
handed off to some server, so that a client can retire) is also better, as in
case of server crash or link failure, it is enough that a single alternative
server is available and reachable, as opposed to a majority of servers.

A significant characteristic of our approach is that it focuses on addressing
network failures and temporary node failures, while not addressing permanent
node failures, leaving them as an orthogonal issue to be attacked
locally, e.g., through storage redundancy at each node.
By not conflating temporary and permanent node failures, our approach does not
impose on the distributed execution the cost of tolerating the latter.

Our approach can be seen to fit in the general philosophy described in
\cite{DBLP:conf/cidr/Helland07} when aiming for ``almost-infinite scaling'':
in avoiding large-scale distributed transactions and only assuming atomic
updates to local durable state; in not requiring exactly-once messaging and
having to cope with duplicates; in using uniquely identified entities; in
remembering messages as state; in having entities manage ``per-partner state''.
Our approach can be seen as applying that philosophy in designing a scalable
distributed data type.

But the CRDT approach that we adopt goes further: since messages are unified
with state, which evolves monotonically with time, the required message
guarantees are even weaker than the at-least-once as assumed in the paper
above.
Messages with what has become an old version of the state need not be
re-transmitted, as the current state includes all relevant information,
subsuming the old state, so it suffices to keep transmitting the current state
to enable progress (assuming that some messages will eventually get through).

\section{Conclusion}
\label{conclusion}

We have addressed the problem of achieving very large scale counters over
unreliable networks. In such scenarios providing strong consistency criteria
precludes availability in general and, even if there are no network
partitions, will impact performance. We have, therefore, defined what we
called \emph{ECDC -- Eventually Consistent Distributed Counters}, that provide
the essence of counting (not losing increments or over-counting), while
giving up the sequencer aspect of the stronger classic distributed counters.

While ECDC can be naively implemented using the CRDT approach, such
implementation is not scalable, as it suffers from what is being perceived to
be the major problem with CRDT approaches: the state pollution with node
id related entries. This pollution involves not only current concurrent
nodes, but also all already retired nodes.

We have presented a solution to ECDC, called \emph{Handoff Counters}, that
adopts the CRDT philosophy, making the ``protocol'' state be a part of the
CRDT state. This allows a clear distinction between what is the durable state
to be preserved after a crash, and what are temporary variables used in the
computation. It also allows a correction assessment to focus on CRDT state and
the merge operation, while allowing a simple distributed gossip algorithm to
be used over an unreliable network (with arbitrary message loss, reordering or
duplication).

Contrary to a naive CRDT based ECDC, our solution achieves scalability in two
ways. First, node id related entries have a local nature, and are not
propagated to the whole distributed system: we can have many thousands of
participating nodes and only a few level 0 entries.  Second, even not
guaranteeing it in the general case, it allows garbage collection of entries
for nodes that participate in the computation and then retire, in normal runs,
while assuring correctness in arbitrary communication patterns. (We have also
sketched an enhancement towards improving garbage collection upon retirement,
which we leave for future work.) These two aspects make our approach usable
for large scale scenarios, contrary to naive CRDT based counters using
client-based ids, and avoiding the availability or reliability problems when
using server-based CRDTs and remote invocation.

Moreover, our approach to overcoming the id explosion problem in CRDTs is not
restricted to counters. As we have discussed, it is more generally applicable
to other data types involving associative and commutative operations.

\bibliographystyle{plain}
\bibliography{counters}

\begin{thebibliography}{10}

\bibitem{AspnesHerlihyShavit1994}
James Aspnes, Maurice Herlihy, and Nir Shavit.
\newblock Counting networks.
\newblock {\em J. ACM}, 41(5):1020--1048, September 1994.

\bibitem{AttiyaDolevWelch1995}
H.~Attiya, S.~Dolev, and J.L. Welch.
\newblock Connection management without retaining information.
\newblock In {\em System Sciences, 1995. Proceedings of the Twenty-Eighth
  Hawaii International Conference on}, volume~2, pages 622 --631 vol.2, jan
  1995.

\bibitem{AttiyaRappoport1994}
Hagit Attiya and Rinat Rappoport.
\newblock The level of handshake required for establishing a connection.
\newblock In Gerard Tel and Paul Vitányi, editors, {\em Distributed
  Algorithms}, volume 857 of {\em Lecture Notes in Computer Science}, pages
  179--193. Springer Berlin Heidelberg, 1994.

\bibitem{Brewer2000}
Eric~A. Brewer.
\newblock Towards robust distributed systems (abstract).
\newblock In {\em Proceedings of the nineteenth annual ACM symposium on
  Principles of distributed computing}, PODC '00, pages 7--, New York, NY, USA,
  2000. ACM.

\bibitem{CerfKahn1974}
V.~Cerf and R.~Kahn.
\newblock A protocol for packet network intercommunication.
\newblock {\em Communications, IEEE Transactions on}, 22(5):637 -- 648, may
  1974.

\bibitem{davey2002}
B.A. Davey and H.A. Priestley.
\newblock {\em Introduction to lattices and order}.
\newblock Cambridge university press, 2002.

\bibitem{Dynamo}
Giuseppe DeCandia, Deniz Hastorun, Madan Jampani, Gunavardhan Kakulapati,
  Avinash Lakshman, Alex Pilchin, Swaminathan Sivasubramanian, Peter Vosshall,
  and Werner Vogels.
\newblock Dynamo: amazon's highly available key-value store.
\newblock In {\em Proceedings of twenty-first ACM SIGOPS symposium on Operating
  systems principles}, SOSP '07, pages 205--220, New York, NY, USA, 2007. ACM.

\bibitem{Gossip1987}
Alan Demers, Dan Greene, Carl Hauser, Wes Irish, John Larson, Scott Shenker,
  Howard Sturgis, Dan Swinehart, and Doug Terry.
\newblock Epidemic algorithms for replicated database maintenance.
\newblock In {\em Proceedings of the sixth annual ACM Symposium on Principles
  of distributed computing}, PODC '87, pages 1--12, New York, NY, USA, 1987.
  ACM.

\bibitem{FeketeLynchMansourSpinelli1993}
Alan Fekete, Nancy Lynch, Yishay Mansour, and John Spinelli.
\newblock The impossibility of implementing reliable communication in the face
  of crashes.
\newblock {\em J. ACM}, 40(5):1087--1107, November 1993.

\bibitem{GilbertLynch2002}
Seth Gilbert and Nancy Lynch.
\newblock Brewer's conjecture and the feasibility of consistent, available,
  partition-tolerant web services.
\newblock {\em SIGACT News}, 33(2):51--59, June 2002.

\bibitem{GoodmanVernonWoest1989}
James~R. Goodman, Mary~K. Vernon, and Philip~J. Woest.
\newblock Efficient synchronization primitives for large-scale cache-coherent
  multiprocessors.
\newblock In {\em Proceedings of the third international conference on
  Architectural support for programming languages and operating systems},
  ASPLOS-III, pages 64--75, New York, NY, USA, 1989. ACM.

\bibitem{CASSANDRA4775}
Arya Goudarzi.
\newblock Cassandra-4775: Counters 2.0.
\newblock \url{https://issues.apache.org/jira/browse/CASSANDRA-4775}, 2012.

\bibitem{GrayLamport2006}
Jim Gray and Leslie Lamport.
\newblock Consensus on transaction commit.
\newblock {\em ACM Trans. Database Syst.}, 31(1):133--160, March 2006.

\bibitem{DBLP:conf/cidr/Helland07}
Pat Helland.
\newblock Life beyond distributed transactions: an apostate's opinion.
\newblock In {\em CIDR}, pages 132--141. www.cidrdb.org, 2007.

\bibitem{HerlihyLimShavit1995}
Maurice Herlihy, Beng-Hong Lim, and Nir Shavit.
\newblock Scalable concurrent counting.
\newblock {\em ACM Trans. Comput. Syst.}, 13(4):343--364, November 1995.

\bibitem{HerlihyWing1990}
Maurice~P. Herlihy and Jeannette~M. Wing.
\newblock Linearizability: a correctness condition for concurrent objects.
\newblock {\em ACM Trans. Program. Lang. Syst.}, 12(3):463--492, July 1990.

\bibitem{Klophaus2010}
Rusty Klophaus.
\newblock Riak core: building distributed applications without shared state.
\newblock In {\em ACM SIGPLAN Commercial Users of Functional Programming}, CUFP
  '10, pages 14:1--14:1, New York, NY, USA, 2010. ACM.

\bibitem{Cassandra}
Avinash Lakshman and Prashant Malik.
\newblock Cassandra: a decentralized structured storage system.
\newblock {\em SIGOPS Oper. Syst. Rev.}, 44(2):35--40, April 2010.

\bibitem{CASSANDRA2495}
Sylvain Lebresne.
\newblock Cassandra-2495: Add a proper retry mechanism for counters in case of
  failed request.
\newblock \url{https://issues.apache.org/jira/browse/CASSANDRA-2495}, 2011.

\bibitem{LynchTuttle1987}
Nancy~A. Lynch and Mark~R. Tuttle.
\newblock Hierarchical correctness proofs for distributed algorithms.
\newblock In {\em Proceedings of the sixth annual ACM Symposium on Principles
  of distributed computing}, PODC '87, pages 137--151, New York, NY, USA, 1987.
  ACM.

\bibitem{ParkerPopek1983}
Jr. Parker, D.S., G.J. Popek, G.~Rudisin, A.~Stoughton, B.J. Walker, E.~Walton,
  J.M. Chow, D.~Edwards, S.~Kiser, and C.~Kline.
\newblock Detection of mutual inconsistency in distributed systems.
\newblock {\em Software Engineering, IEEE Transactions on}, SE-9(3):240 -- 247,
  may 1983.

\bibitem{rep:syn:sh138}
Marc Shapiro, Nuno Pregui{\c c}a, Carlos Baquero, and Marek Zawirski.
\newblock A comprehensive study of {C}onvergent and {C}ommutative {R}eplicated
  {D}ata {T}ypes.
\newblock Rapport de recherche 7506, Institut Nat.\ de la Recherche en
  Informatique et Automatique (INRIA), Rocquencourt, France, January 2011.

\bibitem{Shapiro2011CRDT}
Marc Shapiro, Nuno Pregui\c{c}a, Carlos Baquero, and Marek Zawirski.
\newblock Conflict-free replicated data types.
\newblock In {\em Proceedings of the 13th international conference on
  Stabilization, safety, and security of distributed systems}, SSS'11, pages
  386--400, Berlin, Heidelberg, 2011. Springer-Verlag.

\bibitem{ShavitZemach1996}
Nir Shavit and Asaph Zemach.
\newblock Diffracting trees.
\newblock {\em ACM Trans. Comput. Syst.}, 14(4):385--428, November 1996.

\bibitem{Stone1984}
Harold~S. Stone.
\newblock Database applications of the fetch-and-add instruction.
\newblock {\em IEEE Trans. Comput.}, 33(7):604--612, July 1984.

\bibitem{Terry:1994:SGW:645792.668302}
Douglas~B. Terry, Alan~J. Demers, Karin Petersen, Mike Spreitzer, Marvin
  Theimer, and Brent~W. Welch.
\newblock Session guarantees for weakly consistent replicated data.
\newblock In {\em Proceedings of the Third International Conference on Parallel
  and Distributed Information Systems}, PDIS '94, pages 140--149, Washington,
  DC, USA, 1994. IEEE Computer Society.

\bibitem{Vogels2009}
Werner Vogels.
\newblock Eventually consistent.
\newblock {\em Commun. ACM}, 52(1):40--44, January 2009.

\bibitem{WattenhoferWidmayer2004}
Rogert Wattenhofer and Peter Widmayer.
\newblock The counting pyramid: an adaptive distributed counting scheme.
\newblock {\em J. Parallel Distrib. Comput.}, 64(4):449--460, April 2004.

\bibitem{YewTzengLawrie1987}
Pen-Chung Yew, Nian-Feng Tzeng, and D.H. Lawrie.
\newblock Distributing hot-spot addressing in large-scale multiprocessors.
\newblock {\em IEEE Transactions on Computers}, 36(4):388--395, 1987.

\end{thebibliography}
\end{document}